\documentclass{article}
\title{\bf Generalized Reed-Muller codes over Galois rings}
\author{Harinaivo ANDRIATAHINY$^{(1)}$\\e-mail: hariandriatahiny@gmail.com\\Desiré Arsène RATAHIRINJATOVO$^{(2)}$\\e-mail: ratahirinjatovo@gmail.com\\Sanni José ANDRIANALISEFA$^{(3)}$\\e-mail: sunnysfat@gmail.com\\\\$^{1,2,3}$Mention: Mathematics and Computer Science,\\Domain: Sciences and Technologies,\\University of Antananarivo, Madagascar}

\usepackage{amsmath,amssymb}
\usepackage{amscd,amsfonts}
\usepackage{amsthm}
\usepackage[english]{babel}
\usepackage[T1]{fontenc}
\usepackage[ansinew]{inputenc}
\swapnumbers
\theoremstyle{plain}
\newtheorem{thm}{Theorem}[section]
\newtheorem{prop}[thm]{Proposition}

\newtheorem{cor}[thm]{Corollary}

\theoremstyle{definition}

\newtheorem{rem}[thm]{Remark}

\DeclareMathOperator{\card}{Card}

\DeclareMathOperator{\rank}{Rank}
\DeclareMathOperator{\supp}{Supp}

\begin{document}

\maketitle

\begin{abstract}
Recently, Bhaintwal and Wasan studied the Generalized Reed-Muller codes over the prime power integer residue ring. In this paper, we give a generalization of these codes to Generalized Reed-Muller codes over Galois rings.
\end{abstract}

\section{Introduction}
In 1994, Hammons et al.\cite{HK} showed that some non-linear binary codes with very good parameters are images, under the Gray map, of some linear codes over $\mathbb{Z}_{4}$, the ring of integers modulo $4$. This led to the active study of codes over rings. Many properties of these codes are discovered. And various aspects of coding are dealt in the general setting of Galois rings instead of finite fields.\\
Linear codes over $\mathbb{Z}_{p^s}$ have been studied by several authors. Borges et al.\cite{BF} defined the Reed-Muller code over $\mathbb{Z}_{4}$, called the Quaternary Reed-Muller code, and showed that this code has similar properties as a Reed-Muller code over a finite field. Many authors had theoretical interest in codes over residue rings, and more generally over Galois rings.\\
Recently, Bhaintwal and Wasan \cite{BW1} treated the Generalized Reed-Muller (GRM) codes over $\mathbb{Z}_{p^s}$ for a prime power $p^s$. Our purpose is to present a generalization of GRM codes over $\mathbb{Z}_{p^s}$ to GRM codes over Galois rings.\\
This paper is organized as follows. In section 2, we give a multivariate approach for the GRM codes over Galois rings. In section 3, the standard generator matrix for a GRM code over a Galois ring is given. In section 4, we prove that the images of the GRM codes over Galois rings under the projection map are the usual GRM codes over finite fields. In section 5, the rank of a GRM code over a Galois ring is given. The section 6 will examine the trace descriptions of the Kerdock codes over Galois rings and the GRM codes over Galois rings. In section 7, the dual of a GRM code over a Galois ring is presented. And in section 8, the minimum distance of a GRM code over a Galois ring is provided.

\section{Multivariate representation}
Throughout this paper, we consider the Galois ring $GR(p^s,r)$ of characteristic $p^s$ and cardinality $p^{sr}$, where $p$ is a prime number and $s,r$ are integers $\geq 1$. $GR(p^s,r)$ is a local ring with maximal ideal $pGR(p^s,r)$ and residue field $GR(p^s,r)/pGR(p^s,r)=\mathbb{F}_{q}$, where $q=p^r$.\\
Let $h(x)\in GR(p^s,r)[x]$ be a monic basic primitive polynomial of degree $m\geq 1$ dividing $x^{q^m-1}-1$ and having $\xi$ as a root of order $q^m-1$ in the Galois ring $GR(p^s,r)[x]/(h(x))=GR(p^s,rm)$. $\xi$ is called a primitive element of $GR(p^s,rm)$. Let $n=q^m-1$ and
\begin{equation*}
\mathcal{T}_{m}=\{0,1,\xi,\xi^2,\ldots,\xi^{n-1}\}.
\end{equation*}
$\{1,\xi,\xi^2,\ldots,\xi^{m-1}\}$ is a basis of the free module $GR(p^s,rm)$ of rank $m$ over $GR(p^s,r)$, and we have $GR(p^s,rm)=GR(p^s,r)[\xi]$.\\
We fix the notations $\mathcal{R}=GR(p^s,rm)$ and $L=GR(p^s,r)$.\\
Each element $\xi^{i}\in \mathcal{T}_{m}$ can uniquely be expressed as

\begin{equation}\label{xipuis}
\xi^{i}=b_{1i}+b_{2i}\xi+b_{3i}\xi^2+\ldots+b_{mi}\xi^{m-1},
\end{equation}
where \;$b_{ji}\in L$, $0\leq i\leq n-1$, $1\leq j\leq m$.\\
We adopt the convention\; $\xi^{\infty}=0$.\\
Let
\begin{equation*}
b_i=(b_{1i},b_{2i},b_{3i},\ldots,b_{mi})\;,\; 0\leq i\leq n-1,
\end{equation*}
and\\
$b_{\infty}=(0,0,\ldots,0)$.\\
Let $X$ be the set of variables $x_1,x_2,\ldots,x_m$ and let $P(X)$ be a polynomial in these variables with coefficients from $L$. The degree of a nonzero monomial $x_1^{i_1}x_2^{i_2}\ldots x_m^{i_m}$ is $\sum_{k=1}^{m}i_k$ and the degree of a polynomial $P(X)$ is the largest degree of a monomial in $P(X)$. We define $\deg(0)=-\infty$.\\
We define the evaluation map
\begin{equation}\label{evalmap}
\begin{aligned}
ev:\quad L[X]&\longrightarrow L^{q^m}\\
       P(X)&\longmapsto
       (P(b_{\infty}),P(b_{0}),P(b_{1}),\ldots,P(b_{n-1}))
\end{aligned}
\end{equation}
Let
\begin{equation*}
S=\{P(X)\in L[X] \mid \deg_{x_i}(P(X))\leq q-1\;,\; 1\leq i\leq m\}.
\end{equation*}
Let $\nu$ be an integer such that $0\leq\nu\leq m(q-1)$. Then the $\nu$th order Generalized Reed-Muller code of length $q^m$ over $L$ is defined by
\begin{equation}
RM_{L}(\nu,m)=\{ev(P(X))\mid P(X)\in S\;,\; \deg(P(X))\leq\nu\}.
\end{equation}
The shortened Generalized Reed-Muller code of length $q^m-1$ and order $\nu$ over $L$ denoted by $RM_{L}(\nu,m)^{-}$ is the code obtained from $RM_{L}(\nu,m)$ by puncturing at the first position.

\section{Standard generator matrix}
The component-wise product of any two elements
$\textbf{u}=(u_{0},u_{1},\ldots,u_{n})$ and $\textbf{v}=(v_{0},v_{1},\ldots,v_{n})$ of $L^{n+1}$
is defined by
\begin{equation}\label{comprod}
\textbf{u}\textbf{v}=(u_{0}v_{0},u_{1}v_{1},\ldots,u_{n}v_{n}).
\end{equation}
By (\ref{xipuis}), let us consider the $(m+1)\times q^m$ matrix
\begin{equation*}
G:=\begin{pmatrix}
   1 & 1 & 1 & 1 & ... & 1 \\
   0 & 1 & \xi & \xi^{2} & ... & \xi^{n-1} \
   \end{pmatrix}.
\end{equation*}
$G$ can be expressed as
\begin{equation}\label{genmat}
G:=\begin{pmatrix}
   1 & 1 & 1 & 1 & ... & 1 \\
   0 & b_{10} & b_{11} & b_{12} & ... & b_{1n-1} \\
   0 & b_{20} & b_{21} & b_{22} & ... & b_{2n-1} \\
   ... & ... & ... &... & ... & ... \\
   0 & b_{m0} & b_{m1} & b_{m2} & ... & b_{mn-1} \
   \end{pmatrix}.
\end{equation}
The ith row of $G$ is denoted by $\textbf{v}_{i}$\;,\; $0\leq i\leq m$. Thus, the $\textbf{v}_{i}$ are $q^m$-tuples over $L$. In particular, $\textbf{v}_{0}$ is the all one tuple $1^{q^m}$. Let $\nu$ be an integer such that $0\leq\nu\leq m(q-1)$.\\
From section 2, each $P(X)\in S$ can be expressed as
\begin{equation*}
P(X)=\displaystyle{\sum_{0\leq i_{j}\leq
q-1}}a_{i_{1},...,i_{m}}x_{1}^{i_{1}}x_{2}^{i_{2}}\ldots
x_{m}^{i_{m}},
\end{equation*}
where $a_{i_{1},...,i_{m}}\in GR(p^{s},r)$.\\
By (\ref{evalmap}),(\ref{comprod}) and (\ref{genmat}), we have
\begin{equation*}
ev(x_1^{i_1}x_2^{i_2}\ldots x_m^{i_m})=\textbf{v}_1^{i_1}\textbf{v}_2^{i_2}\ldots \textbf{v}_m^{i_m}.
\end{equation*}
And since the map $ev$ is linear, we have
\begin{equation*}
ev(P(X))=\displaystyle{\sum_{0\leq i_{j}\leq
q-1}}a_{i_{1},...,i_{m}}\textbf{v}_{1}^{i_{1}}\textbf{v}_{2}^{i_{2}}\ldots
\textbf{v}_{m}^{i_{m}}
\end{equation*}
Then, the $\nu$th order Generalized Reed-Muller code $RM_{L}(\nu,m)$ of length $q^m$ over $L$ is defined to be the code generated by all tuples of the form
\begin{equation}\label{vectbase}
\textbf{v}_1^{i_1}\textbf{v}_2^{i_2}\ldots \textbf{v}_m^{i_m}\;,\;0\leq i_j\leq q-1\;,\; 1\leq j\leq m\;,\; \sum_{j=1}^{m}i_j\leq\nu.
\end{equation}
$RM_{L}(0,m)$ is the repetition code of length $q^m$ over $L$.\\
Let $G_{\nu}$ be the matrix whose rows consist of all tuples in (\ref{vectbase}). $G_{\nu}$ is called the standard generator matrix of $RM_{L}(\nu,m)$. The coordinates of any tulpe in $G_{\nu}$ are numbered $\infty,0,1,\ldots,n-1$.

\section{Projection map}
Recall that $q=p^r$, $n=q^m-1$ and $L=GR(p^s,r)$. Since $L/pL=\mathbb{F}_q$, consider the projection map which is defined by reduction modulo $p$
\begin{equation}
\begin{aligned}
\alpha:\quad L&\longrightarrow \mathbb{F}_{q}\\
       a&\longmapsto
       \bar{a}
\end{aligned}
\end{equation}
This map is extended to
\begin{equation*}
\begin{aligned}
\alpha:\quad L[x]&\longrightarrow \mathbb{F}_{q}[x]\\
       f(x)=\sum_{i}a_{i}x^i&\longmapsto
       \bar{f}(x)=\sum_{i}\bar{a}_{i}x^i
\end{aligned}
\end{equation*}
and
\begin{equation}\label{alfa}
\begin{aligned}
\alpha:\quad L^{q^m}&\longrightarrow (\mathbb{F}_{q})^{q^m}\\
       \textbf{v}=(a_{0},a_{1},\ldots,a_{n})&\longmapsto
       \bar{\textbf{v}}=(\bar{a}_{0},\bar{a}_{1},\ldots,\bar{a}_{n})
\end{aligned}
\end{equation}
\begin{prop}
We have $\alpha(RM_{L}(\nu,m))=RM_{\mathbb{F}_q}(\nu,m)$ where $RM_{\mathbb{F}_q}(\nu,m)$ is the usual GRM code of order $\nu$ ($0\leq\nu\leq m(q-1)$) and of length $q^m$ over the finite field $\mathbb{F}_q$\cite{BM}.
\end{prop}
\begin{proof}
By (\ref{alfa}) and (\ref{vectbase}), we have
\begin{equation*}
\alpha(\textbf{v}_1^{i_1}\textbf{v}_2^{i_2}\ldots \textbf{v}_m^{i_m})=\bar{\textbf{v}}_1^{i_1}\bar{\textbf{v}}_2^{i_2}\ldots \bar{\textbf{v}}_m^{i_m}
\end{equation*}
and the $q^m$-tuples
\begin{equation*}
\bar{\textbf{v}}_1^{i_1}\bar{\textbf{v}}_2^{i_2}\ldots \bar{\textbf{v}}_m^{i_m}
\end{equation*}
where $0\leq i_j\leq q-1$, $1\leq j\leq m$, $\sum_{j=1}^{m}i_j\leq\nu$ form a basis for the GRM code $RM_{\mathbb{F}_q}(\nu,m)$. Hence the result.
\end{proof}

\section{Rank}
Consider the $q^m$-tuples
\begin{equation}\label{totalbase}
\textbf{v}_{1}^{i_1}\textbf{v}_2^{i_2}\ldots \textbf{v}_m^{i_m}\;,\; 0\leq i_j\leq q-1\;,\; 1\leq j\leq m.
\end{equation}
\begin{prop}
The $q^m$-tuples in (\ref{totalbase}) form a basis for the free $L$-module $L^{q^m}$ where $L=GR(p^{s},r)$.
\end{prop}
\begin{proof}
Let $\bar{\textbf{v}}_i$ be the image of $\textbf{v}_i$ in $(\mathbb{F}_q)^{q^m}$, $0\leq i\leq m$. From the theory of GRM codes over finite fields, we know that the vectors
\begin{equation*}
\bar{\textsc{}\textbf{v}}_{1}^{i_1}\bar{\textbf{v}}_2^{i_2}\ldots \bar{\textbf{v}}_m^{i_m}\;,\; 0\leq i_j\leq q-1\;,\; 1\leq j\leq m
\end{equation*}
form a basis for $(\mathbb{F}_q)^{q^m}$ over $\mathbb{F}_q$.\\
Let $\textbf{v}\in L^{q^m}$. Then $\bar{\textbf{v}}\in (\mathbb{F}_q)^{q^m}$, and there exist constants $a_{i_{1},...,i_{m}}^{(0)}\in \mathbb{F}_q$ such that
\begin{equation*}
\bar{\textbf{v}}=\displaystyle{\sum_{0\leq i_{j}\leq
q-1}}a_{i_{1},...,i_{m}}^{(0)}\bar{\textbf{v}}_{1}^{i_{1}}\bar{\textbf{v}}_{2}^{i_{2}}\ldots
\bar{\textbf{v}}_{m}^{i_{m}}.
\end{equation*}
Then we have
\begin{equation*}
\textbf{v}=\displaystyle{\sum_{0\leq i_{j}\leq
q-1}}a_{i_{1},...,i_{m}}^{'(0)}\textbf{v}_{1}^{i_{1}}\textbf{v}_{2}^{i_{2}}\ldots
\textbf{v}_{m}^{i_{m}}+p\textbf{u}_1
\end{equation*}
for some $\textbf{u}_1\in L^{q^m}$ and $a_{i_{1},...,i_{m}}^{'(0)}\in L$.\\
There exist constants $a_{i_{1},...,i_{m}}^{'(1)}\in L$ such that
\begin{equation*}
\textbf{u}_1=\displaystyle{\sum_{0\leq i_{j}\leq
q-1}}a_{i_{1},...,i_{m}}^{'(1)}\textbf{v}_{1}^{i_{1}}\textbf{v}_{2}^{i_{2}}\ldots
\textbf{v}_{m}^{i_{m}}+p\textbf{u}_2
\end{equation*}
for some $\textbf{u}_2\in L^{q^m}$.\\
Continuing in this way and noting that $p^s=0$ in $L$, we get constants\\ $a_{i_{1},...,i_{m}}^{'(2)},\ldots,a_{i_{1},...,i_{m}}^{'(s-1)}\in L$ such that
\begin{equation*}
\textbf{v}=\displaystyle{\sum_{0\leq i_{j}\leq q-1}}(a_{i_{1},...,i_{m}}^{'(0)}+pa_{i_{1},...,i_{m}}^{'(1)}+\ldots+p^{s-1}a_{i_{1},...,i_{m}}^{'(s-1)})\textbf{v}_{1}^{i_{1}}\textbf{v}_{2}^{i_{2}}\ldots
\textbf{v}_{m}^{i_{m}}.
\end{equation*}
Hence, each $\textbf{v}\in L^{q^m}$ can be expressed as a linear combination of the tuples $\textbf{v}_{1}^{i_{1}}\textbf{v}_{2}^{i_{2}}\ldots\textbf{v}_{m}^{i_{m}}$\;,\; $0\leq i_{j}\leq q-1$\;,\; $1\leq j\leq m$. Since these tuples are $q^m$ in number and $L^{q^m}$ is a free module of rank $q^m$ over $L$, they must form a basis for $L^{q^m}$.
\end{proof}
\begin{thm}
Let $m$ be a positive integer such that $rm\geq s$ and $q=p^r$. Then, for $0\leq\nu\leq m(q-1)$, the GRM code $RM_{L}(\nu,m)$ is a free $L$-module of rank $k$, where
\begin{equation*}
k=\sum_{i=0}^{\nu}\sum_{j=0}^{m}(-1)^j\binom{m}{j}\binom{i-jq+m-1}{i-jq}.
\end{equation*}
\end{thm}
\begin{proof}
By (\ref{vectbase}), the elements of the set
\begin{equation}
B=\{\textbf{v}_{1}^{i_{1}}\textbf{v}_{2}^{i_{2}}\ldots\textbf{v}_{m}^{i_{m}}\mid 0\leq i_j\leq q-1\;,\; \sum_{j=1}^{m}i_j\leq\nu\}
\end{equation}
span the GRM code $RM_{L}(\nu,m)$.
Since $B$ is a subset of the set $\{\textbf{v}_{1}^{i_{1}}\textbf{v}_{2}^{i_{2}}\ldots\textbf{v}_{m}^{i_{m}}\mid 0\leq i_j\leq q-1\}$ which forms a basis for the free $L$-module $L^{q^m}$, then $B$ must be linearly independent. Thus, $B$ is a basis for $RM_{L}(\nu,m)$ and hence $RM_{L}(\nu,m)$ is a free module over $L$. The images of the elements in $B$ under the map $\alpha$ generate the GRM code $RM_{\mathbb{F}_q}(\nu,m)$ over $\mathbb{F}_q$. Also, these images are linearly independent over $\mathbb{F}_q$. Therefore, the elements $\bar{\textbf{v}}$, where $\textbf{v}\in B$, form a basis for $RM_{\mathbb{F}_q}(\nu,m)=\alpha(RM_{L}(\nu,m))$, and we have $\rank RM_{L}(\nu,m)=\rank RM_{\mathbb{F}_q}(\nu,m)$. It is known from the theory of GRM codes over finite fields that
\begin{equation*}
\rank RM_{\mathbb{F}_q}(\nu,m)=\sum_{i=0}^{\nu}\sum_{j=0}^{m}(-1)^j\binom{m}{j}\binom{i-jq+m-1}{i-jq}.
\end{equation*}
Hence the result.
\end{proof}
\begin{rem}\label{cellnumber}
The rank of the GRM code $RM_{L}(\nu,m)$ is just the number of ways we can place $\nu$ or fewer objects in $m$ cells where no cell is to contain more than $q-1$ objects.
\end{rem}

\section{Trace descriptions}
Each element $c\in\mathcal{R}=GR(p^{s},rm)$ has a unique $p$-adic representation
\begin{equation}
c=\xi_{0}+p\xi_{1}+p^2\xi_{2}+\ldots+p^{s-1}\xi_{s-1}
\end{equation}
where $\xi_{0},\xi_{1},\xi_{2},\ldots,\xi_{s-1}\in \mathcal{T}_{m}=\{0,1,\xi,\xi^2,\ldots,\xi^{n-1}\}$. Under this representation, the Frobenius automorphism is defined by
\begin{equation*}
\begin{aligned}
f:\quad \mathcal{R}&\longrightarrow \mathcal{R}\\
        c=\xi_{0}+p\xi_{1}+p^2\xi_{2}+\ldots+p^{s-1}\xi_{s-1}&\longmapsto
       c^{f}=\xi_{0}^q+p\xi_{1}^q+p^2\xi_{2}+\ldots+p^{s-1}\xi_{s-1}^q
\end{aligned}
\end{equation*}
where $q=p^r$.
$f$ is an automorphism of $\mathcal{R}$, fixes only elements of $L=GR(p^{s},r)$, and generates the group automorphism of $\mathcal{R}$, which is cyclic of order $m$. Note that when $s=1$, $f$ is the usual Frobenius automorphism for $\mathbb{F}_{q^m}$.\\
The relative trace map is defined by
\begin{equation*}
\begin{aligned}
T:\quad \mathcal{R}&\longrightarrow L\\
        c&\longmapsto
       T(c)=c+c^{f}+c^{f^2}+\ldots+c^{f^{m-1}}.
\end{aligned}
\end{equation*}
$T$ is a linear transformation over $L$.

\subsection{Kerdock codes over L}
Let $m$ be a positive integer such that $rm\geq s$ and $n=q^m-1$ with $q=p^r$. Let $h(x)\in L[x]$ be a monic basic primitive polynomial of degree $m$ dividing $x^n-1$ and having $\xi$ as a root of order $n$ in $\mathcal{R}$. Let $g(x)$ be the reciprocal polynomial of $(x^n-1)/((x-1)h(x))$. The shortened Kerdock code $\mathcal{K}^{-}$ is the cyclic code of length $q^m-1$ over $L=GR(p^{s},r)$ with generator polynomial $g(x)$.\\
Since $g(x)\mid x^n-1$, $\mathcal{K}^{-}$ is a free cyclic code of rank $n-\deg g(x)=m+1$ over $L$. A generator matrix of $\mathcal{K}^{-}$ is
\begin{equation*}
G^{-}:=\begin{pmatrix}
   1 & 1 & 1 & ... & 1 \\
   1 & \xi & \xi^{2} & ... & \xi^{n-1} \
   \end{pmatrix}.
\end{equation*}
The Kerdock code $\mathcal{K}$ of length $q^m$ over $L$ is obtained by adding an overall parity-check to $\mathcal{K}^{-}$.\\
Since $rm\geq s$ and $\sum_{i=0}^{n-1}\xi^i=0$, the zero-sum check for the first row of $G^{-}$ is $1$ and for the second row, it is $0$. Thus, a generator matrix for $\mathcal{K}$ is
\begin{equation*}
G:=\begin{pmatrix}
   1 & 1 & 1 & 1 & ... & 1 \\
   0 & 1 & \xi & \xi^{2} & ... & \xi^{n-1} \
   \end{pmatrix}
\end{equation*}
where the elements in the second row of $G$ are considered to be $m$-tuples over $L$. Thus we have $\mathcal{K}=RM_{L}(1,m)$.
\begin{thm}\label{tracedesc}
Let $m$ be a positive integer such that $rm\geq s$ and $n=q^m-1$ with $q=p^r$. Let $\xi$ be a primitive element of $\mathcal{R}=GR(p^{s},rm)$. Then $\mathcal{K}^{-}$ and $\mathcal{K}$ have the following trace descriptions over $\mathcal{R}$\\
1. $\mathcal{K}^{-}=\{\epsilon 1^n+v^{(\lambda)}\mid \epsilon\in L\;,\;\lambda\in\mathcal{R}\}$, where $1^n$ is the all one tuple of length $n$ and
\begin{equation*}
v^{(\lambda)}=(T(\lambda),T(\lambda \xi),T(\lambda\xi^2),\ldots,T(\lambda\xi^{n-1})),
\end{equation*}
2. $\mathcal{K}=\{\epsilon 1^{n+1}+u^{(\lambda)}\mid \epsilon\in L\;,\;\lambda\in\mathcal{R}\}$, where
\begin{equation*}
u^{(\lambda)}=(0,T(\lambda),T(\lambda \xi),T(\lambda\xi^2),\ldots,T(\lambda\xi^{n-1})).
\end{equation*}
\end{thm}
\begin{proof}
1. Let $\mathcal{C}=\{\epsilon 1^n+v^{(\lambda)}\mid \epsilon\in L\;,\;\lambda\in\mathcal{R}\}$. Let $h(x)$ be the monic basic primitive polynomial of degree $m$ in $L(x)$ dividing $x^n-1$ such that $h(\xi)=0$. Let $h^{*}(x)$ be the reciprocal polynomial of $h(x)$, i.e. $h^{*}(x)=x^{m}h(\frac{1}{x})$. From the definition of $\mathcal{K}^{-}$, the check polynomial of $\mathcal{K}^{-}$ is $(1-x)h^{*}(x)$. Clearly, $1-x$ annihilates the tuple $\epsilon 1^n$ and $h^{*}(x)$ annihilates $v^{(\lambda)}$. Thus $(1-x)h^{*}(x)$ annihilates $\mathcal{C}$. It follows that $\mathcal{C}\subseteq \mathcal{K}^{-}$. On the other hand, we have $\card \mathcal{C}=\card \mathcal{K}^{-}=(p^{sr})^{m+1}$.\\
2. $\mathcal{K}$ is a parity check extension of $\mathcal{K}^{-}$ and the zero-sum check for $\epsilon 1^n$ is $\epsilon$ and the zero-sum check for $v^{(\lambda)}$ is $0$. Thus, the result follows from 1.
\end{proof}

\subsection{GRM codes over L}
Let $k$ be any integer such that $0\leq k\leq q^m-1$, where $q=p^r$. Then, $k$ can uniquely be expressed as
\begin{equation*}
k=\sum_{i=0}^{m-1}a_{i}q^{i}\;\;,\;\;0\leq a_i\leq q-1.
\end{equation*}
The $q$-weight of $k$ is defined by
\begin{equation*}
w_{q}(k)=\sum_{i=0}^{m-1}a_{i}.
\end{equation*}
From Remark \ref{cellnumber}, we have also
\begin{rem}
$\rank RM_{L}(\nu,m)=\card (\{j\mid 0\leq j\leq q^m-2\;,\;w_{q}(j)\leq\nu\})$.
\end{rem}
\begin{thm}\label{tracedescri}
Let $m$ be a positive integer such that $rm\geq s$ and $n=q^m-1$ with $q=p^r$. Let $1\leq\nu\leq m(q-1)$. Then, $RM_{L}(\nu,m)$ is generated by the repetition code $RM_{L}(0,m)$ together with all $q^m$-tuples of the form
\begin{equation}\label{vectrace}
(0,T(\lambda_{j}),T(\lambda_{j} \xi^{j}),T(\lambda_{j}\xi^{2j}),\ldots,T(\lambda_{j}\xi^{(n-1)j}))
\end{equation}
where $j$ ranges over a system of representatives of those cyclotomic cosets modulo $q^m-1$ for which $w_{q}(j)\leq\nu$ and $\lambda_{j}$ ranges over $\mathcal{R}=GR(p^{s},rm)$.
\end{thm}
\begin{proof}
Let $\mathcal{C}$ be the code generated by the repetition code $RM_{L}(0,m)$ together with all tuples in (\ref{vectrace}). Since the matrix $G$ in (\ref{genmat}) is a generator matrix for the Kerdock code $\mathcal{K}$, then from Theorem \ref{tracedesc}, for each row $\textbf{v}_{j}$, $j=1,2,\ldots,m$ of $G$, there exists a unique $\lambda_{j}\in\mathcal{R}$ such that
\begin{equation*}
\textbf{v}_{j}=(0,T(\lambda_{j}),T(\lambda_{j} \xi^{j}),T(\lambda_{j}\xi^{2j}),\ldots,T(\lambda_{j}\xi^{(n-1)j})).
\end{equation*}
Thus, the lth coordinate of a tuple
\begin{equation*}
\textbf{v}_{1}^{i_{1}}\textbf{v}_{2}^{i_{2}}\ldots\textbf{v}_{m}^{i_{m}}\;,\; 0\leq i_j\leq q-1\;,\;1\leq j\leq m\;,\; \sum_{j=1}^{m}i_j\leq\nu
\end{equation*}
in the standard generator matrix of the GRM code $RM_{L}(\nu,m)$ is of the form
\begin{equation*}
T(\lambda_{1}z)^{i_1}T(\lambda_{2}z)^{i_2}\ldots T(\lambda_{m}z)^{i_m}\;,\;0\leq i_j\leq q-1\;,\; \sum_{i=1}^{m}i_{j}\leq\nu,
\end{equation*}
where $z=\xi^{l}$ with $\xi^{\infty}=0$. Now for a particular $i$ and $j$, we have
\begin{equation*}
\begin{aligned}
T(\lambda_{i}z)T(\lambda_{j}z)&=\sum_{u=0}^{m-1}(\lambda_{i}z)^{f^u}\sum_{v=0}^{m-1}(\lambda_{j}z)^{f^v}\\
&=\sum_{u=0}^{m-1}\lambda_{i}^{f^u}z^{q^u}\sum_{v=0}^{m-1}\lambda_{j}^{f^v}z^{q^v}\\
&=T(\lambda_{i}\lambda_{j}z^2)+T(\lambda_{i}\lambda_{j}^{f}z^{1+q})+T(\lambda_{i}\lambda_{j}^{f^2}z^{1+q^2})+\ldots\\&+T(\lambda_{i}\lambda_{j}^{f^{m-1}}z^{1+q^{m-1}})\\
&=\sum_{k=0}^{m-1}T(\lambda_{i}\lambda_{j}^{f^{k}}z^{1+q^{k}})=\sum_{t}T(\mu_{t}z^t).
\end{aligned}
\end{equation*}
where $t=1+q^k$ and $\mu_{t}=\lambda_{i}\lambda_{j}^{f^k}$, $0\leq k\leq m-1$. In general, if $\sum_{j=1}^{m}i_j=a\geq 1$, then
\begin{equation*}
T(\lambda_{1}z)^{i_1}T(\lambda_{2}z)^{i_2}\ldots T(\lambda_{m}z)^{i_m}=\sum_{t}T(\mu_{t}z^t),
\end{equation*}
where $t=1+q^{j_1}+q^{j_2}+\ldots+q^{j_{a-1}}$\;,\;$0\leq j_k\leq m-1$\;,\;$1\leq k\leq a-1$, and $\mu_t$ is the corresponding product of the powers of $\lambda_{1},\lambda_{2},\ldots,\lambda_{m}$.\\
It is easy to see that in this expansion of any $T(\lambda_{1}z)^{i_1}T(\lambda_{2}z)^{i_2}\ldots T(\lambda_{m}z)^{i_m}$, the corresponding powers $t$ of $z$ are some representatives of cyclotomic cosets modulo $q^m-1$ and $w_q(t)\leq\sum_{k=1}^{m}i_k$. It follows that each tuple $\textbf{v}_{1}^{i_{1}}\textbf{v}_{2}^{i_{2}}\ldots\textbf{v}_{m}^{i_{m}}$ in the generator matrix of the GRM code $RM_{L}(\nu,m)$ is a linear combination of the all one tuple $1^{q^m}$ and the tuples $(0,T(\lambda_{j}),T(\lambda_{j} \xi^{j}),T(\lambda_{j}\xi^{2j}),\ldots,T(\lambda_{j}\xi^{(n-1)j}))$, where $j$ ranges over a set of coset representatives modulo $q^m-1$ with $w_q(j)\leq\sum_{k=1}^{m}i_k\leq\nu$ and $\lambda_{j}$ ranges over $\mathcal{R}$. Hence $RM_{L}(\nu,m)\subseteq\mathcal{C}$.\\
Conversely, let $\mathcal{C}^{-}$ be the code obtained from $\mathcal{C}$ by puncturing at the first position. That is, $\mathcal{C}^{-}$ is generated by the all one tuple $1^n$ together with all tuples of the form $(T(\lambda_{j}),T(\lambda_{j} \xi^{j}),T(\lambda_{j}\xi^{2j}),\ldots,T(\lambda_{j}\xi^{(n-1)j}))$, where $j$ and $\lambda_{j}$ are as in (\ref{vectrace}).\\
Since $w_q(j)\leq\nu$, it is easy to verify that all these generators are annihilated by the polynomial
\begin{equation*}
f_{\nu}^{*}(x)=(1-x)\prod_{\stackrel{1\leq j\leq q^m-2}{w_q(j)\leq\nu}}(1-\xi^{j}x)
\end{equation*}
where $f_{\nu}^{*}(x)$ is the reciprocal polynomial of
\begin{equation*}
f_{\nu}(x)=(x-1)\prod_{\stackrel{1\leq j\leq q^m-2}{w_q(j)\leq\nu}}(x-\xi^{j})=\prod_{\stackrel{0\leq j\leq q^m-2}{w_q(j)\leq\nu}}(x-\xi^{j})
\end{equation*}
Let $g_{\nu}(x)$ be the reciprocal polynomial to the polynomial
\begin{equation*}
g_{\nu}^{*}(x)=\frac{x^{q^m-1}-1}{f_{\nu}(x)}=\prod_{\stackrel{1\leq j\leq q^m-2}{w_q(j)>\nu}}(x-\xi^{j})
\end{equation*}
and denote by $\mathcal{C}_{\nu}=(g_{\nu}(x))$ the $L$-cyclic code generated by $g_{\nu}(x)$. Then, $f_{\nu}^{*}(x)$ is the check polynomial of $\mathcal{C}_{\nu}$. Therefore, $\mathcal{C}^{-}\subseteq \mathcal{C}_{\nu}$. Thus, $RM_{L}(\nu,m)^{-}\subseteq\mathcal{C}^{-}\subseteq \mathcal{C}_{\nu}$. Clearly,
\begin{equation*}
g_{\nu}(x)=\prod_{\stackrel{1\leq j\leq q^m-2}{w_q(j)>\nu}}(1-\xi^{j}x).
\end{equation*}
We have
\begin{equation*}
\begin{aligned}
\rank \mathcal{C}_{\nu}&=q^m-1-\deg g_{\nu}(x)\\
&=\card(\{j\mid 0\leq j\leq q^m-2\;,\;w_q(j)\leq\nu\})\\
&=\rank RM_{L}(\nu,m)\\
&=\rank RM_{L}(\nu,m)^{-}.
\end{aligned}
\end{equation*}
It follows that $RM_{L}(\nu,m)^{-}=\mathcal{C}^{-}=\mathcal{C}_{\nu}$.
\end{proof}
\begin{cor}\label{polygen}
$RM_{L}(\nu,m)^{-}$ is a $L$-cyclic code generated by the polynomial
\begin{equation}\label{polgener}
g_{\nu}(x)=\prod_{\stackrel{1\leq j\leq q^m-2}{w_q(j)\leq m(q-1)-\nu-1}}(x-\xi^{j}).
\end{equation}
\end{cor}
\begin{proof}
By the proof of Theorem \ref{tracedescri}, we have $RM_{L}(\nu,m)^{-}=(g_{\nu}(x))$ with
\begin{equation*}
g_{\nu}(x)=\prod_{\stackrel{1\leq j\leq q^m-2}{w_q(j)>\nu}}(1-\xi^{j}x).
\end{equation*}
The zeros of $g_{\nu}(x)$ are all $\xi^{-j}$ with $w_q(j)>\nu$. Since $\xi^{-j}=\xi^{q^m-1-j}$ and $w_q(q^m-1-j)=m(q-1)-w_q(j)$, we have
\begin{equation*}
g_{\nu}(x)=\prod_{\stackrel{1\leq j\leq q^m-2}{w_q(j)>\nu}}(x-\xi^{-j})=\prod_{\stackrel{1\leq j\leq q^m-2}{w_q(j)>\nu}}(x-\xi^{q^m-1-j}).
\end{equation*}
Let $J=q^m-1-j$. We have $J\neq0$ and $w_q(J)=m(q-1)-w_q(j)$. Thus, $w_q(j)=m(q-1)-w_q(J)>\nu$. This implies that $w_q(J)<m(q-1)-\nu$. Then
\begin{equation*}
g_{\nu}(x)=\prod_{0<w_q(J)\leq m(q-1)-\nu-1}(x-\xi^{J}).
\end{equation*}
Finally, we have
\begin{equation*}
g_{\nu}(x)=\prod_{\stackrel{1\leq j\leq q^m-2}{w_q(j)\leq m(q-1)-\nu-1}}(x-\xi^{j}).
\end{equation*}
\end{proof}

\section{Dual code}
\begin{prop}\label{zerosum}
Let $m$ be a positive integer such that $rm\geq s$ and $n=q^m-1$ with $q=p^r$. Let $\nu$ be an integer such that $0\leq\nu<m(q-1)$. Then, for any $c=(c_{\infty},c_0,c_1,\ldots,c_{n-1})\in RM_{L}(\nu,m)$, we have $c_{\infty}+c_0+c_1+\ldots+c_{n-1}=0$.
\end{prop}
\begin{proof}
It is enough to prove our Proposition for all generators of the GRM code $RM_{L}(\nu,m)$ given in Theorem \ref{tracedescri}. First, since $rm\geq s$, then for $1^{q^m}$, we have $\underbrace{1+1+\ldots+1}_{q^m terms}=q^m=p^{rm}=0$. Second, for the $q^m$-tuples (\ref{vectrace}), we have
\begin{equation*}
\begin{aligned}
\sum_{i=0}^{n-1}T(\lambda_{j}\xi^{ij})&=\sum_{i=0}^{n-1}\sum_{k=0}^{m-1}(\lambda_{j}\xi^{ij})^{f^k}\\
&=\sum_{k=0}^{m-1}\lambda_{j}^{f^k}\sum_{i=0}^{n-1}\xi^{ijq^k}\\
&=\sum_{k=0}^{m-1}\lambda_{j}^{f^k}\frac{1-\xi^{jnq^k}}{1-\xi^{jq^k}}=0
\end{aligned}
\end{equation*}
since $\xi^{n}=1$.
\end{proof}
\begin{thm}
Let $m$ be a positive integer such that $rm\geq s$ and $0\leq\nu<m(q-1)$ with $q=p^r$. Then
\begin{equation}
RM_{L}(\nu,m)^{\bot}=RM_{L}(\mu,m)
\end{equation}
where $\mu=m(q-1)-\nu-1$.
\end{thm}
\begin{proof}
First, we prove that the all one $q^m$-tuple $1^{q^m}\in RM_{L}(\nu,m)$ belongs to $RM_{L}(\nu,m)^{\bot}$. Since $rm\geq s$, then $1^{q^m}\cdot 1^{q^m}=0$. Moreover, we have to prove that $1^{q^m}$ is orthogonal to all $q^m$-tuples of the form (\ref{vectrace}), where $w_q(j)\leq \nu$. By the proof of Proposition \ref{zerosum}, we have
\begin{equation*}
\sum_{i=0}^{n-1}T(\lambda_{j}\xi^{ij})=0.
\end{equation*}
Therefore, $1^{q^m}\in RM_{L}(\nu,m)^{\bot}$.\\
Next, we prove that any $c=(c_{\infty},c_0,c_1,\ldots,c_{n-1})\in RM_{L}(\mu,m)$ belongs to $RM_{L}(\nu,m)^{\bot}$. Clearly, $c\in RM_{L}(\mu,m)$ if and only if $c-c_{\infty}1^{q^m}\in RM_{L}(\mu,m)$, and $c\in RM_{L}(\nu,m)^{\bot}$ if and only if $c-c_{\infty}1^{q^m}\in RM_{L}(\nu,m)^{\bot}$. Therefore, it is sufficient to show that for $c$ with $c_{\infty}=0$, $c\in RM_{L}(\mu,m)$ implies $c\in RM_{L}(\nu,m)^{\bot}$.\\
Let $c=(0,c')\in RM_{L}(\mu,m)$, where $c'=(c_0,c_1,\ldots,c_{n-1})$.\\Then $c'\in RM_{L}(\mu,m)^{-}$. By Corollary \ref{polygen}, $RM_{L}(\mu,m)^{-}$ is a $L$-cyclic code with generator polynomial
\begin{equation}
g_{\mu}(x)=\prod_{\stackrel{1\leq j\leq q^m-2}{w_q(j)\leq \nu}}(x-\xi^{j}).
\end{equation}
So, $c'(x)=c_0+c_1x+\ldots+c_{n-1}x^{n-1}$ is a multiple of $g_{\mu}(x)$. By Proposition \ref{zerosum}, $c'(1)=c_0+c_1+\ldots+c_{n-1}=0$. Then $c'(x)$ is also a multiple of $x-1$.\\
Since $\bar{g}_{\mu}(1)\neq 0$, $g_{\mu}(1)$ is an invertible element of $L$. It follows that $c'(x)$ is a multiple of $(x-1)g_{\mu}(x)$. Then, $c'(x)$ is annihilated by the polynomial
\begin{equation*}
f_{\mu}(x)=\frac{x^{n}-1}{(x-1)g_{\mu}(x)}=\prod_{\stackrel{1\leq j\leq q^m-2}{w_q(j)>\nu}}(x-\xi^{j})
\end{equation*}
i.e. $c'(x)f_{\mu}(x)=0$.\\
Therefore $c'(x)$ belongs to the dual code of the $L$-cyclic code with generator poltnomial
\begin{equation*}
\begin{aligned}
f_{\mu}^{*}(x)&=\prod_{\stackrel{1\leq j\leq q^m-2}{w_q(j)> \nu}}(1-\xi^{j}x)=\prod_{\stackrel{1\leq j\leq q^m-2}{w_q(j)\leq m(q-1)-\nu-1}}(x-\xi^{j})\\
&=g_{\nu}(x).
\end{aligned}
\end{equation*}
By Corollary \ref{polygen}, $RM_{L}(\nu,m)^{-}=(g_{\nu}(x))$. Hence, $c'(x)\in (RM_{L}(\nu,m)^{-})^{\bot}$. Since $c_{\infty}=0$, $c\in RM_{L}(\nu,m)^{\bot}$. Therefore, we have proved that $RM_{L}(\mu,m)\subseteq RM_{L}(\nu,m)^{\bot}$.\\
To check the ranks, we first note that $\rank RM_{L}(\nu,m)=\rank RM_{L}(\nu,m)^{-}$ as $RM_{L}(\nu,m)$ is a parity check extension of $RM_{L}(\nu,m)^{-}$. By Corollary \ref{polygen}, $RM_{L}(\nu,m)^{-}$ has the generator polynomial $g_{\nu}(x)$ as given in (\ref{polgener}), and\\$RM_{L}(\mu,m)^{-}$ has the generator polynomial $g_{\mu}(x)$. $g_{\mu}(x)$ is the reciprocal polynomial of $\frac{x^n-1}{(x-1)g_{\nu}(x)}$. Therefore, we have
\begin{equation*}
\begin{aligned}
\rank RM_{L}(\nu,m)^{-}+\rank RM_{L}(\mu,m)^{-}&=(n-\deg g_{\nu}(x))+(n-\deg g_{\mu}(x))\\
&=n+1=q^m.
\end{aligned}
\end{equation*}
It follows that $\rank RM_{L}(\nu,m)+\rank RM_{L}(\mu,m)=q^m$.\\
Thus, $\rank RM_{L}(\mu,m)=\rank RM_{L}(\nu,m)^{\bot}$.\\
And we have $RM_{L}(\mu,m)=RM_{L}(\nu,m)^{\bot}$.
\end{proof}

\section{Minimum distance}
\begin{thm}
The shortened GRM code $RM_{L}(\nu,m)^{-}$ is a subcode of a BCH code of length $q^m-1$ over $L$ whose roots include $\xi,\xi^2,\ldots,\xi^{(R+1)q^Q-2}$ where $\xi$ is a primitive element of $\mathcal{R}=GR(p^s,rm)$, and $Q$ and $R$ are the quotient and remainder respectively, resulting from dividing $\mu+1=m(q-1)-\nu$ by $q-1$.
\end{thm}
\begin{proof}
Let $d$ be the smallest integer such that $w_q(d)=m(q-1)-\nu=(q-1)Q+R$\;,\;$0\leq R<q-1$. Therefore, we must have $d=Rq^Q+(q-1)q^{Q-1}+(q-1)q^{Q-2}+\ldots+(q-1)q+(q-1)=(R+1)q^{Q}-1$.\\
Also, every integer less than $d$ has $q$-weight less than or equal to $m(q-1)-\nu-1$. It follows from (\ref{polgener}) that all elements $\xi,\xi^2,\ldots,\xi^{(R+1)q^Q-2}$ are roots of $RM_{L}(\nu,m)^{-}$. Thus, $RM_{L}(\nu,m)^{-}$ is a subcode of a primitive BCH code of length $q^m-1$ over $L$.
\end{proof}
Consequently, from BCH bound on codes over Galois rings \cite{BW}, $RM_{L}(\nu,m)^{-}$ has minimum distance at least $(R+1)q^Q-1$, where $Q$ and $R$ are the quotient and remainder respectively, resulting from dividing $\mu+1=m(q-1)-\nu$ by $q-1$. It can be easily seen from the structure of $RM_{L}(\nu,m)$ that the minimum distance of $RM_{L}(\nu,m)$ is equal to the minimum distance of $RM_{L}(\nu,m)^{-}$. Hence the minimum distance of $RM_{L}(\nu,m)$ is at least $(R+1)q^Q-1$.
\begin{thm}
The GRM code $RM_{L}(\nu,m)$ has minimum distance $(R+1)q^Q-1$, where $Q$ and $R$ are the quotient and remainder respectively, resulting from dividing $\mu+1=m(q-1)-\nu$ by $q-1$.
\end{thm}
\begin{proof}
since the minimum distance of $RM_{L}(\nu,m)$ is at least $(R+1)q^Q-1$, we only need to show a tuple of weight $(R+1)q^Q-1$ in $RM_{L}(\nu,m)$. The image $\alpha(RM_{L}(\nu,m))=RM_{\mathbb{F}_q}(\nu,m)$ has minimum distance exactly $(R+1)q^Q-1$. Let $\textbf{u}=(u_{\infty},u_{0},u_{1},\ldots,u_{n-1})$ be a vector of weight $(R+1)q^Q-1$ in $RM_{\mathbb{F}_q}(\nu,m)$. Let $I=\supp (\textbf{u})=\{i\mid u_i\neq 0\}$ the support of $\textbf{u}$. Then, $\card(I)=(R+1)q^Q-1$. Then, there exists a vector $\textbf{v}=(v_{\infty},v_{0},v_{1},\ldots,v_{n-1})\in RM_{L}(\nu,m)$ such that $\alpha(\textbf{v})=\bar{\textbf{v}}=\textbf{u}$ i.e. $(\bar{v}_{\infty},\bar{v}_{0},\bar{v}_{1},\ldots,\bar{v}_{n-1})=(u_{\infty},u_{0},u_{1},\ldots,u_{n-1})$. Thus, $\bar{v}_{i}=u_{i}$ for all $i$.\\
If $i\notin I$, then $u_{i}=0$. Thus, $v_{i}$ is in $pL$, and $p^{s-1}v_{i}=0$.\\
If $i\in I$, then $u_{i}\neq 0$ and $v_{i}\notin pL$, i.e. $v_{i}$ is an invertible element of $L$, and $p^{s-1}v_{i}\neq 0$. Therefore, $p^{s-1}\textbf{v}\in RM_{L}(\nu,m)$ and $p^{s-1}\textbf{v}$ is of weight $(R+1)q^Q-1$. Hence, $RM_{L}(\nu,m)$ has minimum distance $(R+1)q^Q-1$.
\end{proof}

\end{document}